\documentclass[preprint,12pt]{elsarticle}

\usepackage{amssymb}
\usepackage{extarrows}

\newtheorem{thm}{Theorem}
\newtheorem{lem}[thm]{Lemma}
\newtheorem{prop}[thm]{Proposition}
\newtheorem{exam}[thm]{Example}
\newdefinition{rmk}{Remark}
\newproof{proof}{Proof}
\newproof{pot}{Proof of Theorem \ref{thm2}}

\def\mP{{\mathbb P}}

\journal{Elsevier}

\begin{document}

\begin{frontmatter}



\title{A new conception for computing gr\"{o}bner basis and its applications}


\author{Lei Huang}

\address{Key Laboratory of Mathematics Mechanization\\
 Institute of Systems Science,   AMSS\\
 Beijing 100190,  China\\
 lhuang@mmrc.iss.ac.cn}

\begin{abstract}
This paper presents a conception for computing gr\"{o}bner basis. We convert some of gr\"{o}bner-computing algorithms, e.g.,  F5, extended F5 and GWV algorithms into a special type of algorithm. The new algorithm's finite termination problem can be described by equivalent conditions, so all the above algorithms can be determined when they terminate finitely. At last, a new criterion is presented. It is an improvement for the Rewritten and Signature Criterion.
\end{abstract}

\begin{keyword}
 \sep Gr\"obner basis \sep F5 \sep GVW \sep TRB \sep Mpair



\end{keyword}

\end{frontmatter}


\section{Introduction}
Since the Gr\"{o}bner basis was proposed from 1965 (Buchberger, \cite{buchberger1}), it has been implemented in most computer algebra systems (e.g., Maple,
Mathematica, Magma, Sage, Singular, Macaulay 2, CoCoA, etc).

There has been extensive effort in finding more efficient algorithms for computing
Gr\"{o}bner bases. e.g., Buchberger \cite{buchberger2, buchberger3}, Lazard (1983, \cite{lazard}), Moller, Mora and Traverso (1992,
\cite{moller}), Faug\`{e}re (1999, \cite{faugere1}). In 2002, Faug\`{e}re presented the F5 algorithm to detect useless S-polynomials by the Syzygy and Rewritten criterions \cite{faugere2}. This algorithm had the fastest speed for a long time. It was also discussed and improved by many papers; see Eder and Perry (2009,
\cite{eder}), Sun and Wang (2009, \cite{wang1}), Hashemi and Ars (2010, \cite{hashemi}). Hashemi and Ars extended the F5 algorithm by modifying the signature order. This modification can bring more efficiency to the F5 algorithm. Recently, Gao, Volny IV and Wang (2010, \cite{gao1, gao2}) proposed new conceptions and techniques to compute Gr\"{o}bner basis. e.g., they proposed the conception of pairs; generalized the signature order to be an arbitrary one; used arbitrary top reductions to instead F5 reductions; etc.

For the greater efficiency, new techniques were described more and more complicate than before. It constitutes obstacles for people to understand all points of algorithms, to make comparisons between different algorithms, and to search for new algorithms. e.g., the finite termination problem. This problem can be easily determined for simple algorithms. But for recently proposed algorithms, it becomes not easy. Faug\`{e}re, Hashemi and Ars tried to prove the F5's finite termination problem in their paper \cite{faugere2, hashemi} with a few lines, but few people could understand their proofs clearly. In September 2010, Gao, Volny IV and Wang announced at their paper (see \cite{gao2}) that the termination of GVW algorithm is a open problem; They also believed that the same problem of F5 has not be solved yet.

The author studied the termination problems of these algorithms. Some results (F5's and GVW's) were discovered in different ways, and described by different languages. To prove them together, we need to summarize their common points to build a general algorithm. By absorbed ideas from the GVW and F5B (see Sun and Wang \cite{wang2}) algorithms, we built the general TRB algorithm, where 'TRB' comes from the fact: all these algorithms have a common purpose of generating TRB pairs. In particular, the TRB algorithm has the following features:
\begin{itemize}
\item
Some efficient algorithms can be converted into regular TRB algorithms.
\item
Problems can be discussed together with the TRB algorithm, e.g., the correctness and termination problems.
\item
It provides a platform for generating new algorithms.
\end{itemize}

In this paper, we proved all the F5, extended F5, and GVW algorithms are regular TRB algorithms. With a general discussion, their terminations were all described. The conclusion is: F5 and extended F5 algorithms always terminate finitely. The GVW algorithm has finite termination if the monomial order and signature order are almost compatible.

The last topic is a new criterion (Mpair Criterion) for detecting useless S-polynomials. The new criterion can block more unnecessary pairs than before.  We proved that the Rewritten and Signature Criterions can hardly block more pairs than Mpair Criterion. And sometimes unnecessary pairs meet only the new proposed criterion.

This paper is organized as follows: In Section 2, we introduce some basic conceptions. The definition and a correctness proof of TRB algorithm are proposed. In Section 3, 4 and 5, we convert the F5, Extended F5 and GVW algorithms into regular TRB algorithms respectively. Section 6 provides some equivalent conditions for the termination problem of TRB algorithm. In Section 7, we propose the Mpair Criterion.

\section{Comments and Definitions}
Let $\mathbb K$ be a field, $\mP={\mathbb K}[x_1, \cdots, x_n]$ the polynomial ring, $f=(f_1, \cdots, f_d)\in \mP^d$ the initial polynomials list.

$(u,v)\in \mP^d\times \mP$ is a \textbf{pair} if $u \cdot f=v$ holds. Denote $PAIR \subset\mP^d\times \mP$ the set of all pairs.

The monomials' set and terms' set of $\mathbb P$ are denoted by $M$ and $T$ respectively. $s\in\mP^d$ is called a monomial (term) if $s=mE_i$, where $m$ is a monomial (term) of $\mP$, $E_i$ is the $i$-th canonical unit vector, $i=1,\cdots,d$. The set of all monomials (terms) in $\mP^d$ is denoted by $ME$ ($TE$).

There are three main orders to be used in this paper. The \textbf{monomial order} $\prec_m$, \textbf{signature order} $\prec_s$, \textbf{pair order} $\prec_p$ are defined over $\mP$, $\mP^d$ and $PAIR$ respectively. $\prec_m$ and $\prec_s$ are both admissible orders.  $(u_1,v_1)\prec_p(u_2,v_2)\Leftrightarrow lm(v_1)lm(u_2)\prec_s lm(v_2)lm(u_1)$.

Let $p=(u,v)\in PAIR$. Orders $\prec_m$, $\prec_s$ and $\prec_p$ are applied on $v$, $u$ and $p$ respectively.
The \textbf{leading monomial} (\textbf{leading term}) of pair $p$ is defined same to the  leading monomial (term) of $v$. i.e., $lm(p)=lm(v)$, $lt(p)=lt(v)$.
Define the \textbf{signature} of $p$ as the leading monomial of $u$. $sig(p)=lm(u)$.

\vspace{.3cm}
We call two pairs \textbf{equivalent}, $p_1\equiv p_2$, if $sig(p_1)=sig(p_2)$ and $lm(p_1)=lm(p_2)$. Call them \textbf{similar}, $p_1\sim p_2$, if $lm(p_1) sig(p_2) = lm(p_2) sig(p_1)$.

A pair $p$ is called \textbf{syzygy} if $lm(p)=0$. In the paper, we also call signature $s$ syzygy, if there has a syzygy pair $p$ satisfied $sig(p)=s$.

Say pair $p_1$ is \textbf{top reducible} by pair $p_2$, if both of them are non-syzygy, $p_1\prec_p p_2$ and $lm(p_2)|lm(p_1)$. The corresponding top reduction is to replace pair $p_1$ by $\displaystyle p_1-\frac{lt(p_1)}{lt(p_2)}p_2$.

If non-syzygy pair $p$ can not be top reduced by any pair, call $p$ a \textbf{top reductional prime} pair, simply by \textbf{TRP} pair. All the TRP pairs form the set $TRP$. Put similar-to-$p$ TRP pairs all together to form a set, call it \textbf{TRP similar set} of $p$. We call pair $p'$ is a \textbf{top reductional basis} (\textbf{TRB}) pair if it is a TRP pair, and $sig(p')$ can not be proper divided by any signatures in the same TRP similar set.

A \textbf{multiplied pair} $[m,p]$ is an element in $M\times PAIR$. Its value equals $mp$. $sig([m,p])=sig(mp)$. Call $p_1$ is top reducible by multiplied pair $[m,p]$, if $p_1$ is top reducible by $p$, and $lm(p_1)=lm(mp)$. Given two non-similar and non-syzygy pairs $p_1$ and $p_2$, their \textbf{joint multiplied pair} (\textbf{J-pair}) is defined as $[m,p]$, where $p$ equals either $p_1$ or $p_2$ respectively, when $p_1\prec_p p_2$ or $\succ_p p_2$. $m=\displaystyle \frac{lcm(lm(p_1),lm(p_2))}{lm(p)}$.

\vspace{.3cm}
A property will be usually used in the left of paper. Proved it below before the start of discussions.
\begin{prop}\label{trppro}
For every non-syzygy signature $s$, there will have at least one TRP pair $p$ signed $s$ ($sig(p)=s$). All the pair $p'$ signed $s$ will satisfy $p\preceq_p p'$. $p\sim p'$ if and only if $p'\in TRP$.
\end{prop}
\begin{proof}
\item[1.]
Denote $s=x^{\alpha} E_i$. We start from the pair $p_0=(s,x^\alpha f_i)$, where $f_i$ is the $i$-th initial polynomial. If $p_0$ is top reducible, perform top reduction on $p_0$, to get a new pair $p_1$ signed same signature $s$. $p_1$ will $\prec_p p_0$. If $p_1$ is still top reducible, continue performing top reduction on $p_1$, .... Since $lm(p_i)$ can not always be smaller, we will finally get a top irreducible pair $p$ signed $s$. $p$ is of course not syzygy, and must be a TRP pair.
\item[2.]
For every top reducible pair $p'_0$ signed $s$, we can do the similar things to get another TRP pair signed $s$. Say to get $p'$. $p'\prec_p p'_0$. The remaining thing is to prove $p'\sim p$.
\item[3.]
Suppose $p=(u,v)$ and $p'=(u', v')$ are not similar. We will have $sig(p')=sig(p)$ and $lm(p')\neq lm(p)$. Then $lc(u')p-lc(u)p'$ will top reduce either $p$ or $p'$. Contradiction.
\end{proof}

\noindent\rule{8.5cm}{1pt}
\noindent\textbf{Algorithm (TRB)}

\vspace{.3cm}
\noindent\textbf{Input:} $f=(f_1,\cdots,f_d) \in \mP^d$,

Admissible orders $\prec_m$ and $\prec_s$ over $\mP$ and $\mP^d$ respectively.

\noindent\textbf{Output:} $DONE$, the set stored all the results.

\noindent\rule{8.5cm}{1pt}

\noindent\textbf{Procedures.}

\noindent\textbf{Step 0.}
    Set $DONE:=\phi$,
    $TODO:=\{[1,(E_i,f_i)],~i=1,\cdots,d\}$.

    ~~~~//$TODO$ stores multiplied pairs for the future computation.

\noindent\textbf{Step 1.}\

    \textbf{if} $TODO=\phi$, \textbf{then} output $DONE$.

     \textbf{else} $[m,p]:=Selection(TODO)$.
    ~//$[m,p]$ is a multiplied pair.

    ~~~~$TODO:=TODO\setminus\{[m,p]\}$.

    \textbf{end}

\noindent\textbf{Step 2.}\

    \textbf{if} $Criterions([m,p])=true$, \textbf{then} go to Step 1.

    \textbf{else}  $p':=Reductions(mp)$.

    \textbf{end}

\noindent\textbf{Step 3}\

    \textbf{if} $CheckStore(p')=true$ \textbf{then}

    ~~~~$TODO:=TODO\cup \{Jpair(p', p'')|p''\in DONE\}$,

    ~~~~$DONE:=DONE\cup\{p'\}$.

    \textbf{end}.

    Go to Step 1.

\noindent\rule{8.5cm}{1pt}

\vspace{.5cm}
The related functions are explained below.
\begin{itemize}
\item
\textbf{Selection} is a function to select a multiplied pair out from $TODO$ for the next computation. The selected pair always has the smallest signature w.r.t. $\prec_s$.
\item
\textbf{Criterions} is a function formed by some criterions. It returns true if the selected multiplied pair meets one of these criterions. We say the multiplied pair pass the criterions if it returns false.
\item
\textbf{Reductions} is a function of the composition of a series of top reductions. The output MUST BE top irreducible by any pairs.
\item
\textbf{CheckStore} is a function to detect whether the reduced pair need to be stored. Pair $p'$ will be blocked by this function if it is non-initial, and equivalent to $mp$, where $mp$ is the input of the corresponding Reductions.
\item
\textbf{JPair} is a function to output the J-pair of input pairs. It will return empty if either the input pairs are similar, or one of them is syzygy.
\end{itemize}


Call a TRB algorithm \textbf{regular}, if it satisfies the following:
For all initial polynomials $f$, when the algorithm terminates, all the TRB pairs (up to equivalence) have been computed out and stored to $DONE$.

\begin{thm}
Let $f$ be the initial polynomials. When a regular TRB algorithm terminates, all the polynomials of $DONE$ form a Gr\"{o}bner basis of $f$.
\end{thm}
\begin{proof}
\item[1.]
For each non-zero polynomial $v\in \langle f\rangle$, it corresponds to a non-syzygy pair $p=(u,v)$.
\item[2.]
If $p$ is a TRP pair, with the definition of TRB pair, there will exist TRB pair $p'$ to satisfy $p\equiv m'p'$, where $m'$ is a monomial of $\mP$. Then $lm(p)$ can be reduced by $lm(p')$.
\item[3.]
If $p$ is not a TRP pair,  by the below proposition, $p$ can be top reduced by a TRB pair $p'$. Similarly $lm(p')|lm(p)$.
\end{proof}

\begin{prop}\label{trtrb}
A top reducible pair can always be top reduced by $TRB$.
\end{prop}
\begin{proof}
\item[1.] If there are top reducible pairs can not be top reduced by $TRB$, suppose $p$ has the smallest signature among them.
Say $p$ can be top reduced by $[m_1,p_1]$.
\item[2.]
If $m_1p_1$ is top reducible, since $sig(m_1p_1)\prec_s sig(p)$, by the assumption, $m_1p_1$ can be top reduced by $p_2\in TRB$. So can $p$.
\item[3.]
If $m_1p_1$ is top irreducible by $PAIR$, since $lm(m_1p_1)\neq 0$, $m_1p_1$ should be a TRP pair. Suppose $m_1p_1\equiv [m_2,p_2]\in M\times TRB$. Then $p$ can be top reduced by $p_2$.
\end{proof}

\section{The TRB-F5 Algorithm}
Usually, a TRB algorithm can be implemented by modifying three functions: Criterions, Reductions and CheckStore.
Now define the TRB-F5 algorithm as follows:
\begin{itemize}
\item Consider only homogeneous polynomials. Choose $\prec_m$ to be a homogeneous order. $\prec_s$ is defined as
\begin{equation}\label{f5precs}
x^{\alpha}E_i \prec_s x^{\beta}E_j\Longleftrightarrow\left\{
\begin{array}{l}
i>j \ \ or
\\
i=j, \ x^{\alpha}\prec_m x^{\beta}.
\end{array}
\right.
\end{equation}
\item
\textbf{F5Criterions} is composed by the Syzygy Criterion and Rewritten Criterion.
\item
We say $p_1$ is \textbf{F5 reducible} by $[m_2,p_2]$ if it is top reducible by $[m_2,p_2]$, and $[m_2,p_2]$ passes (does not meet) the F5Criterions.

\textbf{F5Reductions} is the function to perform F5 reductions (by $DONE$) as many as possible until the result can not be F5 reduced by $DONE$.
\item
\textbf{CheckStore-F5} copied the general CheckStore. It blocks the non-initial pair which is equivalent to the input of F5Reductions;
\end{itemize}

We now describe criterions used in the TRB-F5 algorithm. Define $index(p)=i$, where $sig(p)=x^{\alpha}E_i$.
We say $[m,p]$ meets the \textbf{Syzygy Criterion}, if there exists a pair $p_1=(u_1,v_1)\in DONE$, such that $i=index(p)<index(p_1)$ and $lm(v_1)E_i|sig(mp)$.

Define order $\prec_{F5}$ as follows. $p_1\prec_{F5} p_2$ if and only if
$$\left\{
\begin{array}{l}
index(p_1)>index(p_2) \ \ or
\\
index(p_1)=index(p_2), deg(p_1)<deg(p_2),
\end{array}
\right.$$
where $deg(p)=|\alpha|$, if $lm(p)=x^\alpha$.

A \textbf{Rewrite Rule List} (abbreviated by \textbf{Rules}) has been built to describe the Rewritten Criterion. Rules was initialized by empty. Before every Selection (of Step 1.) performed, do something on $[m,p]$, where $[m,p]$ satisfies
\begin{itemize}
\item
$[m,p]$ has the smallest $\prec_5$ order in $TODO$;
\item
$[m,p]$ passes F5Criterions.
\end{itemize}
We will replace all such $[m,p]$ in $TODO$ by $[1,mp]$, and prepend $mp$ at the head of Rules. After $mp$ was reduced by F5Reductions, replace $mp$ in Rules by $F5Reductions(mp)$.

Let $[m,p]$ be a multiplied pair. Find out in Rules the first pair whose signature can divide $sig(mp)$. Say it is $p'$. $[m,p]$ meets the \textbf{Rewritten Criterion} if $p\neq p'$.


\vspace{.3cm}
There are two things need to be determined in the following. One is to prove TRB-F5 is a regular TRB algorithm. The other one is to show that TRB-F5 is in the TRB-language of the F5 algorithm.

In the process of running a TRB algorithm, say $[m,p]$ is the last selected multiplied pair from Selection. Call signature $s$ \textbf{considered}, if $s\prec_s sig(mp)$; Call it \textbf{considering} or \textbf{unconsidered} if $s=$ or $\succ_s sig(mp)$ respectively.
\begin{lem}\label{f5singlev}
For each monomial $m$, at most one multiplied pair $[m_1,p_1]$ can meet all the following conditions:
\begin{enumerate}
\item
$lm(m_1p_1)=m$;
\item
$sig(m_1p_1)$ is considered;
\item
$p_1\in DONE$;
\item
$[m_1,p_1]$ passes F5Criterions.
\end{enumerate}
\end{lem}
\begin{proof}
\item[1.] Suppose there are two multiplied pairs $[m_1,p_1]$ and $[m_2,p_2]$ able to meet the above conditions.
If $p_1\sim p_2$, we have $[m_1,p_1]\equiv[m_2,p_2]$. One of them will meet the Rewritten Criterion. Contradiction. So $p_1\not\sim p_2$.
\item[2.]
    Suppose $[m'_1,p_1]$ is the J-pair of $p_1$ and $p_2$. We have $m'_1|m_1$, because $lcm(lm(p_1),lm(p_2))|m$. Then
$sig(m'_1p_1)\preceq_s sig(m_1p_1)$ is considered. $[m'_1,p_1]$ has been already stored into $TODO$, because $p_1$, $p_2\in DONE$.
\item[3.]
$m'_1\neq 1$. Otherwise $p_1$ can be F5 reduced by $p_2$ and can not be an output of F5Reductions.
\item[4.]
Since $m'_1|m_1$, $[m'_1,p_1]$ can also pass F5Criterions. When $sig(m'_1p)$ becomes the smallest (up to $\prec_{F5}$) in $TODO$, $m'_1p$ is prepended to Rules. Then $[m_1,p_1]$ will be rewritten. Contradiction.
\end{proof}


\begin{prop} \label{f5trb}
Let $p$ be a non-syzygy pair signed a considered signature. Then
\begin{enumerate}
\item If $p$ is top reducible, it can be F5 reduced by $DONE$;
\item If $p$ in Rules, it should also be in $DONE$;
\item If $p\in TRP$, there has $[m_1, p_1]\in M\times DONE$ can pass F5Criterions and satisfies $p\equiv m_1p_1$.
\end{enumerate}
\end{prop}
\begin{proof}
\item[1.] If there have considered non-syzygy pairs can not meet the conclusions, suppose $p_0$ has the smallest signature among them.
\item[2.]
If $p_0$ is top reducible by $[m_1,p_1]$, we have $sig(m_1p_1)\prec_s sig(p_0)$. Then $m_1p_1$ will meet the conclusion. This means there will have $[m_2,p_2]\in M\times DONE$ to satisfy that it passes F5Criterions, $lm(m_2p_2)=lm(m_1p_1)=lm(p_0)$ and $sig(m_2p_2)\preceq_s sig(m_1p_1)\prec_s sig(p_0)$.
\item[3.]
If $p_0$ is in Rules but not in $DONE$, since $sig(p_0)$ is considered, the only possibility is: After J-pair $[m'_0,p'_0]$ changed to $[1,m'_0p'_0]$, $p_0=m'_0p'_0$ could not be F5 reduced by $DONE$. Then it will at last be blocked by the CheckStore.

But by Part 2., this case can hardly happen. Because every J-pair is top reducible, $p_0$ should be F5 reducible by $DONE$.
\item[4.]
If $p_0$ is a TRP pair, suppose $[1,p_0]$ is rewritten by $[m_1,p_1]$. We assert $p_0\sim p_1$. Otherwise, by Proposition (\ref{trppro}), $m_1p_1$ will be top reducible and F5 reducible by $[m_2,p_2]$. Then, $[m_1,p_1]$ and $[m_2,p_2]$ will contradict to Lemma (\ref{f5singlev}).
\end{proof}

The above conclusion told us a lot of things. Suppose $sig(p)$ is considered.
\begin{itemize}
\item
If $p$ is top reducible, it can be F5 reduced by one and only one pair of $DONE$.
\item
Every output of the F5Reductions is top irreducible.
\item
If $p\in TRB$, it will $\equiv p_1\in DONE$. (Only $[1,p_1]$ can meet the third conclusion of Proposition (\ref{f5trb}).)
\item
TRB-F5 is a regular TRB algorithm.
\item
The function CheckStore will always return true.
\end{itemize}

\vspace{.3cm}
The following makes some comparison between the TRB-F5 and F5 algorithms. There is a lot of differences between them. Some large differences are listed below.
\begin{description}
\item[The selection order.] The F5 algorithm selects the smallest multiplied pair w.r.t. order $\prec_{F5}$ from $TODO$ for the next computation.
If two pairs are $\prec_{F5}$ equivalent, select the smaller signature one.

According to considering only homogenous polynomials, the above order is equivalent to the signature order which defined in TRB-F5.
\item[Action 1.] In the F5 algorithm, when the F5Reductions performing, some other multiplied pairs may be stored to $TODO$. Decompose F5Reductions as follows:
    $$p_0\xrightarrow[p'_0\in DONE]{} p_1\xrightarrow[DONE]{} \cdots\xrightarrow[p'_{k-1}\in DONE]{} p_k$$
    where $p_0$ and $p_k$ are the input and output of F5Reductions respectively, $p_0\xrightarrow[p'_1\in DONE]{} p_1$ represents that $p_0$ is F5 reduced into $p_1$ by $p'_1$. In the F5 algorithm, at the moment of each $p_i$ computed out, \textbf{temporary multiplied pairs} $[m,p]$ who satisfy the following conditions will also be stored to $TODO$.
\begin{enumerate}
\item $[m,p]\in M\times DONE$ passes F5Criterions;
\item $lm(mp)=lm(p_i)$;
\item $p\prec_p p_i$.
\end{enumerate}
\item[Action 2.] In the F5 algorithm, $[m,p]$ will be replaced in the following case:
\begin{itemize}
\item $[m,p]$ has the smallest $\prec_{F5}$ order in $TODO$;
\item $[m,p]$ passes F5Criterions;
\item $[m,p]$ is the J-pair of $p$ and $p_1\in DONE$, and $p_1$ can F5 reduce  $mp$.
\end{itemize}
\item[Action 3.]
Instead of $[1,mp]$, the F5 algorithm will replace $[m,p]$ by $[1,mp-m_1p_1]$, and prepend $mp-m_1p_1$ to the Rules.
\end{description}

From the Action 1., temporary multiplied pairs in truth can only generated from the function output, $p_k$. These pairs will be generated again in the Step 3.
\begin{prop}
Suppose $p_i, 0<i<k$, is a middle reduction result of F5Reductions. No multiplied pair can satisfy all the conditions in Action 1.
\end{prop}
\begin{proof}
\item[1.] Suppose $[m,p]$ satisfies $[m,p]\in M\times DONE$ passes the F5Criterions, $lm(mp)=lm(p_i)$ and $p\prec_p p_i$.
We know that $p_i$ is F5 reduced by $[m'_{i},p'_{i}]\in M\times DONE$. $[m,p]$ can also be F5 reduced by $[m'_{i},p'_{i}]$.
\item[2.] Since both $p$ and $p'_{i}\in DONE$, their J-pair $[m',p]$ had be stored to $TODO$, where $m'|m$.

$sig(m'p)\preceq_{F5} p_i$, so $m'p$ has been prepended into Rules. If $m'=1$, $p$ will not be in $DONE$; If $m'\neq 1$, $[m,p]$ will be rewritten.
\end{proof}
%
%

\begin{prop}
Every $[m,p]$ who satisfies the first two conditions of Action 2, will always satisfy the third condition.
\end{prop}
\begin{proof}
\item[1.] By above discussions, we know that J-pair $[m,p]$ is F5 reducible by $DONE$. Say $mp$ can be F5 reduced by $[m_1,p_1]\in M\times DONE$. We will prove that $[m,p]$ should be the J-pair of $p$ and $p_1$.
\item[2.]
    Denote $[m',p]$ the J-pair of $p$ and $p_1$, where $m'|m$. If $m'\neq m$, with the algorithm, $m'p$ was already prepended to Rules, and $[m,p]$ was rewritten. Contradiction.
\end{proof}

For the Action 3., prepending $mp$ or $mp-m_1p_1$ to Rules can hardly bring real difference in the algorithm, because they have a same signature. If $[m,p]$ was replaced by $[1,mp]$, it has been deduced that $mp$ can be F5 reduced by only one pair $p_1$. The algorithm will do this reduction first to change $mp$ to $mp-m_1p_1$.

\section{The Extended F5 algorithm}
An extended F5 algorithm has been proposed since 2010 (see \cite{hashemi}). The main improvement to F5 is modifying the signature order for efficiency.

The \textbf{TRB-EF5} algorithm consists of three functions: \newline \{\textbf{EF5Criterions, F5Reductions, CheckStore}\}. \newline The F5Reductions and CheckStore have been introduced already. We need only define EF5Criterions here.

\textbf{EF5Criterions} is composed of the ESyzygy and ERewritten Criterion.
ESyzygy Criterion is a modification of the Syzygy Criterion to suit new signature orders.
Describe it with the \textbf{Syzygy Signatures Set} (\textbf{Syzygies}). Syzygies stores all the signatures as $lm(p)E_i$, where $i\in {\mathbb N}$, $(E_i,f_i)\prec_p p\in DONE$. $[m_0,p_0]$ meets the \textbf{ESyzygy Criterion} if $sig(m_0p_0)$ is divided by one of signatures in Syzygies.

ERewritten Criterion is proposed to simplify the Rewritten Criterion. Used the ERewritten Criterion, we need not to replace any pairs in $TODO$. $[m,p]$ meets the \textbf{ERewritten Criterion} if and only if there is a pair $p'\in DONE$ satisfied $sig(p')|sig(mp)$ and $sig(p')\succ_s sig(p)$.

The ERewritten Criterion in truth is a special case of the Rewritten Criterion. In the Rewrite Rules of the F5 algorithms, among $\prec_{F5}$ equivalent pairs, it has no requirement for them to be prepended first. The ERewritten Criterion let pairs be prepended with the signature order.

The EF5 algorithm consider also homogeneous polynomials and a homogeneous monomial order $\prec_m$.
Similarly we can prove TRB-EF5 is also a regular TRB algorithm. The main improvement to the F5 algorithm is the modifications of the signature order.
In \cite{hashemi}, two modified signature orders were proposed. They were defined as
$$
x^{\alpha}E_i\prec_s x^{\beta}E_j \Leftrightarrow
 \left\{
\begin{array}{l}
lm(x^{\alpha}f_j)\prec_m lm(x^{\beta}f_i) \ \ or
\\
lm(x^{\alpha}f_j)=lm(x^{\beta}f_i), lm(f_j)\prec_m lm(f_i);
\end{array}
\right.
$$
or
$$x^{\alpha}E_i\prec_s x^{\beta}E_j \Leftrightarrow
 \left\{
\begin{array}{l}
deg(x^\alpha f_i)<deg(x^{\beta}f_j) \ \ or
\\
deg(x^\alpha f_i)=deg(x^{\beta}f_j), x^\alpha\prec_m x^\beta \ \ or
\\
deg(x^\alpha f_i)=deg(x^{\beta}f_j), x^\alpha=x^\beta, i<j.
\end{array}
\right.
$$
With these modifications, the new algorithms experimentally terminated at a lower degree than F5.

\section{The TRB-GVW Algorithm}
The GVW algorithm was presented recently. It is in fact another regular TRB algorithm. Let us define the TRB-GVW algorithm below.

\begin{itemize}
\item $\prec_m$ and $\prec_s$ are arbitrary admissible orders.
\item
\textbf{GVWCriterions} is composed of the GCyzygy Criterion and Signature Criterion.
\item
\textbf{TopReductions:} Do top reductions (by $DONE$) as many as possible until the result can not be top reduced by $DONE$.
\item
Pair $p$ will be blocked by the \textbf{CheckStore-GVW}, if and only if $p\equiv [m_1,p_1]\in M\times DONE$.
\end{itemize}

$[m,p]$ meets the \textbf{GSyzygy Criterion}, if $sig(mp)$ is divided by one of signatures in Syzygies, where \textbf{Syzygies} stores all the syzygy signatures in $DONE$, and all the following signatures:
$max(sig(v_2p_1), sig(v_1p_2))$, where $p_i=(u_i,v_i)\in DONE$, $sig(v_2p_1)\neq sig(v_1p_2)$.

$[m,p]$ meets the \textbf{Signature Criterion}, if another pair $p'$ signed the same signature $sig(mp)$, passed this function already.

\begin{lem}
Let $s$ be a non-initial and TRB signature. There will have a J-pair of two TRB pairs signed $s$.
\end{lem}
\begin{proof}
\item[1.]
All the $E_i$ are TRB signatures, so there have TRB pairs to satisfy that their signatures divide $s$ and $\neq s$. Suppose $p_1$ has the smallest $\prec_p$ order among them. Let $m\in M$ satisfy $sig(m_1p_1)=s$.
\item[2.]
$m_1p_1$ is top reducible, or $sig(m_1p_1)$ is not a TRB signature. Say $m_1p_1$ can be top reduced by $[m_2,p_2]\in M\times TRB$. Denote $[m'_1,p_1]$ the J-pair of $p_1$ and $p_2$, where $m'_1|m_1$.
\item[3.]
Suppose $p_3$ is a TRP pair signed $sig(m'_1p_1)$, $p_3\equiv [m_4,p_4]\in M\times TRB$. Since $m'_1p_1$ is top reducible, we have $p_4\sim p_3\prec_p p_1$ and $sig(p_4)|sig(p_3)|s$. With the definition of $p_1$, $sig(p_4)$ must equal $s$ and $m'_1=m_1$. Then $[m'_1,p_1]$ is what we need.
\end{proof}

\begin{prop}
Let $p$ be a non-syzygy pair signed a considered signature. There will have a multiplied pair $[m_1,p_1]\in M\times DONE$ to satisfy $lm(m_1p_1)=lm(p)$ and $p_1\succeq_p p$. Where $p_1\sim p$ if and only if $p\in TRP$.
\end{prop}
\begin{proof}
\item[1.] If there have pairs to contradict the conclusion, suppose $p_0$ has the smallest signature among them.
\item[2.]
If $p_0$ is top reducible by $[m_1,p_1]$, we have $sig(m_1p_1)\prec_s sig(p_0)$. There is a pair $[m_2,p_2]\in M\times DONE$ to satisfy $lm(m_2p_2)$ $=lm(m_1p_1)= sig(p_0)$ and $sig(m_2p_2)\preceq_s sig(m_1p_1)\prec_s sig(p_0)$.
\item[3.]
Suppose $p_0$ is a TRB pair. We need to prove $p_0\in DONE$. If $[m'_0,p'_0]\in TODO$ signed $sig(p_0)$, when $sig(p_0)$ considering, $p_0$ will be reduced from $m'_0p'_0$ by the TopReductions and stored to $DONE$. So prove $sig(p_0)\in TODO$ is necessary.

If $sig(p_0)$ is initial, $[1,(E_i,f_i)]$ had already been stored. If it is not initial, by the above lemma, there is a J-pair of two TRB pairs, $[m_1,p_1]$, to satisfy $sig(p_0)=sig(m_1p_1)$. By the assumption, both of these two TRB pairs are in $DONE$. And $[m_1,p_1]$ should be in $TODO$.
\item[4.]
If $p_0$ is a TRP and non-TRB pair, we have $p_0\equiv [m_1,p_1]\in M\times TRB$.
\end{proof}

With above propositions, we have the following comments:
\begin{itemize}
\item
The output of Reductions is always top irreducible.
\item
All the TRB pairs up to equivalence have been computed out.
\item
TRB-GVW is a regular TRB algorithm.
\item
Only TRB pairs can be stored.
\end{itemize}

\vspace{.3cm}
The main difference between TRB-GVW and GVW algorithms is: GVW use the regular top reductions (simply by regular reductions) to reduce pairs.  Pair $p_1$ is \textbf{regular reducible} by $p_2$, if and only if
\begin{itemize}
\item
both of them are non-syzygy;
\item $lm(p_2)|lm(p_1)$;
\item $p_1\preceq_p p_2$;
\item $lt(u_1)lt(p_2)\neq lt(u_2)lt(p_1)$, where $p_i=(u_i,v_i)$.
\end{itemize}
Pair $p_1$ top reducible by $p_2$ will deduce it also regular reducible by $p_2$. And the reverse is not always true.
But TRP pairs can not be regular reduced further, because they have already been the smallest pairs.
\begin{prop}
If replace top reductions in the TopReductions by regular reductions, the result will not be changed (up to equivalence).
\end{prop}
\begin{proof}
\item[1.]
From $mp$, suppose we get two different pairs $p_1$ and $p_2$ by the unchanged and changed TopReductions respectively. We have $p_1$ and $p_2$ are both top irreducible. By Proposition (\ref{trppro}), $p_1\sim p_2$.
\item[2.]
In addition, we also have $lc(u_1)lc(p_2)=lc(u_2)lc(p_1)$. Otherwise $p_2$ will be regular reduced further. And we will get a pair smaller than $p_2$ signed the same signature.
\end{proof}

\section{Equivalent conditions for the finite termination of a TRB algorithm.}

The signature order of F5's can be modified for the efficiency: for lower syzygy signatures, for more efficient array eliminations (see the F4 algorithm, \cite{faugere1}), or for some other things. Replaced by the GSyzygy Criterion, every admissible signature order can generate a regular TRB algorithm. A question is: Are they all valuable?

The answer is false, because sometimes the modified algorithms will terminate infinitely. e.g., modifying the signature order by
$$x^{\alpha}E_i\prec_s x^{\beta}E_j  \Longleftrightarrow
\left\{
\begin{array}{l}
lm(x^{\alpha}f_i)\prec_m lm(x^{\beta}f_j) \ \ or
\\
lm(x^{\alpha}f_i)= lm(x^{\beta}f_j), i<j
\end{array}
\right.
$$
will generate an algorithm to terminate finitely; But modifying by
$$x^{\alpha}E_i\prec_s x^{\beta}E_j  \Leftrightarrow
\left\{
\begin{array}{l}
deg(x^{\alpha}f_i)<deg(x^{\beta}f_j) \ \ or
\\
deg(x^{\alpha}f_i)=deg(x^{\beta}f_j), x^\alpha f_i\succ_mx^\beta f_j \ or
\\
x^\alpha f_i=x^\beta f_j, i<j
\end{array}
\right.
$$
will lead to an algorithm of infinite termination. In the rest of this section, we will prove this argument.

\begin{lem}\label{monomials}
$S=(m_1, m_2,\cdots)$ is an infinite monomials sequence, where $0\neq m_i\in M$. $S$ always
has an infinite subsequence $(m_{k_1}, m_{k_2},\cdots)$ to satisfy that $k_i<k_j$ and $m_{k_i}|m_{k_j}$, for all $i<j$.
\end{lem}
\begin{proof}
\item[1.]
Use mathematical induction on the number of variables, $n$. When
$n=0$, all monomials can only be $1$. The conclusion is directly true.
Suppose $n$ is the smallest number to make the conclusion be not always true.
Denote $m_i=x_1^{\alpha_{1,i}}x_2^{\alpha_{2,i}}\cdots x_n^{\alpha_{n,i}}, \forall i$.
\item[2.]
Suppose there is a monomial $m_k$ to satisfy $m_k\nmid m_i, \forall i>k$.
Define subsequences of $S$, $S_{j,\alpha}$, where $0<j\leq n$ and $0\leq\alpha<\alpha_{j,k}$,
$$
S_{j,\alpha}:=(m_i~|~i>k, \alpha_{j,i}=\alpha).
$$
We assert that at least one of these subsequences has infinite number of elements.
For every $i>k$, since $m_k\nmid m_i$, $m_i$ will belong to at least one of above subsequences. Then we have
$$
\cup S_{j,\alpha}=\{m_i|i>k\}
$$
has infinite elements. But there are only finite number of subsequences, so at least one of them is infinite.

Say $S_{j,\alpha}$ is infinite. All the monomials in $S_{j,\alpha}$ have the same component $x_{j}^{\alpha}$. Ignored it, $S_{j,\alpha}$ will be equivalent to an infinite sequence corresponding to $n-1$ variables. By the induction hypothesis, it satisfies our assumptions.
\item[3.]
For all monomials $m_k$, if there always has a monomial $m_i$ to satisfy $m_k|m_i$ and $i>k$. The subsequence can be built one by one.
\end{proof}

Let $\prec_m$ and $\prec_s$ are admissible orders over $\mP$ and $\mP^d$ respectively. Call them \textbf{compatible} if for all $s_1, s_2\in ME$, $m_1, m_2\in M$, $s_1,s_2,m_1, m_2\neq 0$, it always deduces to $s_1m_1\preceq_s s_2m_2$ from $s_1\preceq_s s_2$ and $m_1\preceq_m m_2$. The equality holds if and only if $s_1=s_2$ and $m_1=m_2$.

Restricted $\prec_s$ onto each $\mP$ branch of $\mP^d$, we will get $d$ distinct sub-orders $\prec_{s,i}$ over $\mP$. $\prec_s$ is compatible to $\prec_m$ if and only if all the $\prec_{s,i}$ are same to $\prec_m$. Call $\prec_m$ and $\prec_s$ \textbf{almost compatible} if they are either compatible, or there has only one sub-order $\prec_{s,k}$ not same to $\prec_m$, and satisfies $x^{\alpha}E_k\prec_s E_i$, for all $\alpha$ and $i\neq k$.

\begin{lem}\label{alcompate}
If $\prec_s$ and $\prec_m$ are not almost compatible, there are $\{i,j,x^\alpha,x^\beta,x^{\gamma}\}\in {\mathbb N}\times {\mathbb N}\times M\times M\times M$ to satisfy
\begin{itemize}
\item
$\gcd(x^\beta, x^{\gamma})=1, x^\beta \prec_{s,i} x^\gamma, x^\gamma\prec_m x^\beta$;
\item
$\gcd(x^\alpha, x^\gamma)=1, x^\alpha E_i \prec_s E_j$.
\end{itemize}
\end{lem}
\begin{proof}
\item[1.] Suppose $\prec_{s,1}$ is not same to $\prec_m$. There are $x^\beta$ and $x^\gamma$ to satisfy $\gcd(x^\beta, x^{\gamma})=1, x^\beta \prec_{s,1} x^\gamma, x^\gamma\prec_m x^\beta$. We need to find out $\alpha$ and $k$ to satisfy $\gcd(x^\alpha, x^\gamma)=1, x^\alpha E_1 \prec_s E_k$.
\item[2.]
    If $E_1\succ_s E_k$ for some $k\neq 1$, $\{1,k,1,x^\beta,x^\gamma\}$ will meet the conclusion. We consider the case $E_1\prec_s E_k$ for all $k\neq 1$.
\item[3.]
All the $\prec_{s,k}$, $k\neq 1$ are same to $\prec_m$, or $\{k,1,1,x^{\beta'},x^{\gamma'}\}$ meets the conclusion.
\item[4.]
By the definition of almost compatible, $x^\alpha E_1\succ_s E_2$, for some $\alpha$.
Let $x^{\alpha}=x_1^{\alpha_1}\cdots x_n^{\alpha_n}$. Denote $X_1$ and $X_2$ by $\displaystyle \prod_{x_i\nmid x^{\gamma}}x_i^{\alpha_i}$ and $\displaystyle \prod_{x_i| x^{\gamma}}x_i^{\alpha_i}$ respectively. $X_1X_2=x^{\alpha}$.
Let $c$ be a positive integer to satisfy $X_2|x^{c\gamma}$. Define $x^{\alpha'}=x^{c\beta}X_1$. We have $\gcd(x^{\alpha'},x^{\gamma})=1$ and $x^{\alpha'}E_1\succ_s E_2$. Otherwise $x^{\alpha'}E_1\prec_s E_2$ will deduce
$$
X_2E_2\succ_s X_2x^{\alpha'}E_1=x^{\alpha+c\beta}E_1\succ_sx^{c\beta}E_2\succeq_sx^{c\gamma}E_2\succeq_sX_2E_2.
$$
At last $\{1,2,x^{\alpha'},x^\beta,x^\gamma\}$ meets the conclusion.
\end{proof}

\begin{thm}
If an algorithm always has finite termination for all input polynomials, call it \textbf{terminated algorithm}. The following conditions are equivalent to each other.
\begin{enumerate}
\item
 The regular TRB algorithm is a terminated algorithm.
\item
For every initial polynomials f, there have only finite number of TRB equivalent sets.
\item
Orders $\prec_m$ and $\prec_s$ are almost compatible.
\item
For every f, their have only finite number of TRP similar sets.
\end{enumerate}
\end{thm}
\begin{proof}
\item[$1.\Rightarrow 2.$]
When the algorithm finished, all the TRB pairs up to equivalence were computed out and stored to $DONE$. They have a finite number.
\item[$2.\Rightarrow 3.$]
If $\prec_s$ and $\prec_m$ are  not almost compatible, we can find $\{i,j,x^\alpha, x^\beta, x^\gamma\}$ to meet the conclusion of Lemma (\ref{alcompate}).
Initialize polynomials as $f_i=x^{\gamma}, f_j=x^{\alpha+\beta}-x^{\alpha+\gamma}, f_k=0, \forall k\neq i,j$. The TRB pairs includes of $(E_i,x^{\gamma})$, $(E_j, x^{\alpha+\beta}-x^{\alpha+\gamma})$ and $(x^{\alpha+t\beta}E_i+\cdots, x^{\alpha+(t+1)\gamma})$, for all $t\geq 1$. All these pairs are not equivalent to each other.
\item[$3.\Rightarrow 4.$]
Decompose the set $TRP$ into $\cup PS_i, i=1,\cdots,d$, where $PS_i$ stores all the TRP pairs corresponding to index $i$. Suppose set $PS_k$ has infinite number of TRP pairs, all these pairs are not similar to each other.

If $\prec_{s,k}$ is not same to $\prec_m$, by the definition, $x^\alpha E_k\prec_s E_i$, for all $\alpha$ and $i\neq k$.
$(E_k, f_k)$ will be the largest TRB pair, and all the TRP pairs with index $k$ will be similar to $(E_k, f_k)$.

Suppose $\prec_{s,k}$ is same to $\prec_m$. By Lemma (\ref{monomials}), we can find out a sequence $(p_{1},p_{2},\cdots)$ from $PS_k$ to satisfy that all the $p_i$ are not similar to each other, $sig(p_{i})|sig(p_{j})$ and $lm(p_{i})|lm(p_{j})$ for all $i<j$.
But this is impossible: Consider only the $p_1$ and $p_2$.

If $p_1 \succ_p p_2$, $p_2$ can be top reduced by $p_1$.

If $p_1\prec_p p_2$, since $\prec_{s,k}$ is same to $\prec_m$, $p_2$ can be top reduced by $\displaystyle p_2-\frac{lt(u_2)}{lt(u_1)}p_1$.
\item[$4.\Rightarrow 1.$]
We present all the TRP similar groups by $p_1, p_2, \cdots p_k$, where $p_i\succ_p p_{i+1}$.

According to the descriptions of TRB algorithm, pairs in $DONE$ are only TRP or Syzygy pairs. The new generated J-pair is similar to the smaller one of its contributed pairs. So, except for the initial ones, every multiplied pair in $TODO$ is similar to a TRP pair.
  Define $N_{TODO}=[n_0, n_1,\cdots, n_{k}]\in {\mathbb N}^{k+1}$, where $n_0$ records the number of initial multiplied pairs in $TODO$, $n_i$ records the number of similar-to-$p_i$ multiplied pairs in $TODO$, $i=1,\cdots,k$.

When $[m,p]\in TODO$ is selected out, by the TRB algorithm, it will be either discarded, or stored to $DONE$ as a new top irreducible pair $p'$, where $p'\prec_p p$. The new generated J-pairs from $p'$ will be $\preceq_p p'\prec_p p$. So, after each loop finished, $N_{TODO}$ will be proper smaller than before w.r.t. the Lexico order.

$N_{TODO}$ can not always be smaller. At last the algorithm will terminated when $N_{TODO}=[0,\cdots,0]$.
\end{proof}

With above result, we propose the conclusions of this section: The F5 algorithm and the extended F5 are both terminated algorithms. The GVW algorithm has finite termination for all input polynomials if and only if the admissible orders $\prec_m$ and $\prec_s$ are almost compatible. In particular, the G2V algorithm (see \cite{gao1}) can always terminate finitely.

\section {Mpair Criterion and the TRB-MJ Algorithm}
Although the GSyzygy Criterion improves Syzygy Criterion, the Signature Criterion is not as powerful as the Rewritten Criterion, because there may have some signatures $s$, all the multiplied pairs in $TODO$ signed $s$ meet the Rewritten Criterion. In this case, we need not to perform reductions on these signatures.

With the proof of Lemma (\ref{f5singlev}), the Rewritten Criterion and the ERewritten Criterion can be improved as a SRewritten Criterion. $[m,p]$ meets the SRewritten Criterion if
\begin{itemize}
\item
There is a pair $p_1\in DONE$, such that $sig(p_1)|sig(mp)$ and $sig(p)\prec_s sig(p_1)$;
\item
or there is a pair $p_2\in DONE$, such that $sig(m_2p_2)=sig(mp)$, $sig(p_2)\equiv_{F5} sig(p)$ and $[m_2,p_2]\not\in TODO$.
\end{itemize}
SRewritten Criterion can block more multiplied pairs. See the following example.
\begin{exam}\label{f5criexam}
Set $\prec_m$ be the Degree Reverse Lexico order, $\prec_s$ defined as (\ref{f5precs}).
Compute the Gr\"{o}bner basis of $$f=\{x_1x_4, x_1x_2-x_2^2, x_1x_3-x_3^2\}.$$
With the TRB-F5 or TRB-GVW algorithm, we will compute TRP pairs one by one as follows:
$(E_3, x_1x_3-x_3^2), (E_2, x_1x_2-x_2^2), (x_3E_2+\cdots, x_2^2x_3-x_3^2x_2), (E_1, x_1x_4),
(x_3E_1+\cdots, x_3^2x_4),$ $(x_2E_1+\cdots, x_2^2x_4), (x_2x_3E_1+\cdots, x_2x_3^2x_4)
$.

The last one is not a TRB pair, it is similar to $(x_3E_1+\cdots,x_3^2x_4)$. It passes ERewritten Criterion but meets the SRewritten Criterion.
\end{exam}

But SRewritten Criterion is just a limited improvement. In the above example, if replace the second initial polynomial by $x_1x_2x_5-x_2^2x_5$, it also can not block the last TRP pair. In this section, a new criterion is proposed to block such unnecessary pairs. We call it Mpair Criterion. A conclusion is: All the non-TRB (and non-syzygy) signatures will meet the Mpair Criterion, so we need only to compute the TRB signatures (and some Syzygy signatures).


Suppose $\prec_m$ and $\prec_s$ over $\mP$ and $\mP^d$ respectively are compatible.
The \textbf{TRB-MJ algorithm} is defined by \newline \{\textbf{MJCriterions, TopReductions, CheckStore}\}.

The \textbf{MJCriterions} consists of the GSyzygy Criterion and the Mpair Criterion.
Call $[m,p]$ meets the \textbf{Mpair Criterion} if $[m,p]$ is neither initial nor an M-pair of $DONE$, where the M-pair is defined below.

Let $PS$ be a pairs set. $[m_0,p_0]\in M\times PS$ is a \textbf{minimal multiplied pair} (\textbf{M-pair}) of $PS$ signed signature $s$, if
\begin{itemize}
\item
 $s=sig(m_0p_0)$,
\item
$m_0\neq 1$,
\item
for all $[m,p]\in M\times PS$ signed $s$, $m\neq 1$, $p\not\equiv p_0$, we always have $p_0\prec_p p$, or $p_0\sim p$ but $sig(p_0)\succ_s sig(p)$.
\end{itemize}

The CheckStore in this algorithm will always return true, because when a J-pair passed the Criterions, it must top reducible by a TRB pair in $DONE$ (This property will be proved later). Then the output of TopReductions will pass the CheckStore. So, all the J-pairs passed the MJCriterions will generate new pairs in $DONE$. Say $[m,p]$ is an \textbf{MJ-pair}, if it is both of J-pair and M-pair of $DONE$. The signature $sig(mp)$ is called \textbf{MJ-signature}. Study MJ-signatures is necessary.

\begin{prop}\label{mjtrb}
All TRB pairs was calculated by the TRB-MJ algorithm.
\end{prop}
\begin{proof}
\item[1.]
Suppose $s$ is a TRB signature. All the TRB signatures smaller than $s$ have been stored to $DONE$ already. If $s$ is initial, it will pass the MJCriterions, and be reduced and stored to $DONE$. We suppose $s$ is not initial below.
\item[2.]
The M-pair of $DONE$ corresponding to $s$ must be exist. Say it is $[m,p]\in M\times DONE$, where $m\neq 1$.
With the definition of TRB pairs, $mp$ is top reducible. By Proposition (\ref{trtrb}), $mp$ can be top reduced by $p_1\in DONE\cap TRB$.
\item[3.]
Denote $[m', p]$ the J-pair of $p$ and $p_1$, where $m'|m$.
Let $p_2$ be a TRP pair signed $sig(m'p)$, $p_2\equiv [m_3,p_3]\in M\times TRB$. Since $p_3\prec_p p$, $sig(p_3)|sig(p_2)|s$, but $p_3$ is not the Mpair, we deduces If $sig(p_3)= s$.
\item[4.]
At last, we have $sig(p_3)=sig(p_2)=s$. This means $s$ is an MJ-signature of $DONE$. $[m,p]$ will pass MJCriterions, be reduced into a TRB pair, and be stored to $DONE$.
\end{proof}

With this proposition, we know that TRB-MJ is also a regular TRB algorithm. The TRB-MJ algorithm has another important property below.
\begin{prop}
All the signatures in $DONE$ can only be either syzygy or TRB signatures.
\end{prop}
\begin{proof}
\item[1.]
Suppose $s$ is a nether syzygy nor TRB signature. All such signatures (neither syzygy nor TRB) smaller than $s$ were rejected from $DONE$. We should to prove that $s$ is also not in $DONE$.
\item[2.]
Let $p$ be a TRP pair signed $s$, $p\equiv [m_1,p_1]\in M\times TRB$, where $m_1\neq 1$. By Proposition (\ref{mjtrb}), $p_1$ has been stored in $DONE$ already. The Mpair signed $s$ should be $\sim p_1\sim p$. Then it can neither be top reduced nor be a J-pair.
\end{proof}

In general, the Mpair Criterion has the following features:
\begin{itemize}
\item
For a same signature $s$, if the M-pair and another multiplied pair are not same. To reduce them into TRP pairs, the M-pair might need less top reductions than the other one, because it has the smallest polynomial part.
\item
With the definition of regular TRB algorithm, the Rewritten or Signature Criterions can hardly block more signatures than the Mpair Criterion. And
sometimes they may miss non-TRB and non-syzygy signatures. So the conclusion of these non-syzygy criterions is:

\textbf{Signature $\prec$ Rewritten} or \textbf{ERewritten\newline  $\prec$ SRewritten $\prec$ Mpair}.
\end{itemize}

\begin{exam}
Back to Example (\ref{f5criexam}) again. When $f_2$ was replaced by $x_1x_2x_5-x_2^2x_5$, the signature $x_2x_3x_5E_1$ will never be reduced by the TRB-MJ algorithm, because the multiplied pair $[x_3,(x_2x_5E_1+\cdots, x_2^2x_4x_5)]$ meets the Mpair Criterion.
\end{exam}

\section{Conclusions and Future}
This paper presented a new conception for the the comparison among some Gr\"{o}bner-computing algorithms. With this conception, we presented the equivalent termination conditions for some important algorithms. We also proposed the Mpair Criterion for the computing.

The author is now looking for techniques to implement TRB algorithms more efficient. Some techniques (in particularly for the TRB-MJ algorithm) will be presented in the next paper.

%
\bibliographystyle{abbrv}

\end{document}